\documentclass[a4paper]{article}
\usepackage[T1]{fontenc}
\usepackage[utf8]{inputenc}
\usepackage{textcomp}
\usepackage{makeidx}
\usepackage{graphicx}
\usepackage{rotating}
\usepackage{float}
\usepackage{changepage}

\usepackage[toctitles]{titlesec}
\usepackage{fancyhdr}

\usepackage{hyperref}
\usepackage{amssymb}
\usepackage{amsmath,bm}
\usepackage{latexsym}
\usepackage{amsthm}
\usepackage{amssymb}
\usepackage{cases}
\usepackage{fancyvrb}

\usepackage{sectsty}
\usepackage{fancybox}
\usepackage{listings}
\usepackage{charter}

\usepackage[titletoc]{appendix}
\usepackage{geometry}

\usepackage{tikz}
\usetikzlibrary{topaths,calc}
\usepackage{enumitem}

\usepackage{epsfig}
\usepackage{tabularx}
\usepackage{slashbox} 
\usepackage{latexsym}
\usepackage{alltt}
\usepackage{setspace}
\usepackage{datetime}
\newdateformat{mydate}{\monthname[\THEMONTH], \THEYEAR}

\usepackage{pslatex}
\usepackage{times}

\usepackage[lined,algonl,boxed]{algorithm2e}

\def\qed{\hfill$\Box$}

\newtheorem{corollary}{Corollary~}
\newtheorem{definition}{Definition~}
\newtheorem{theorem}{Theorem~}
\newtheorem{proposition}{Proposition~}

\newtheorem{observation}{Observation~}

\def\abstract{{\begin{center}
\Large {\bf Abstract}
\end{center} }}

\bmdefine\bchi{\bm\chi}
\DefineVerbatimEnvironment{code}{Verbatim}{fontsize=\small}
\DefineVerbatimEnvironment{example}{Verbatim}{fontsize=\small}
\title{\bf Strong $(r,p)$ Cover for Hypergraphs}
\author{Tapas Kumar Mishra \qquad Sudebkumar Prasant Pal \\ Dept. of Computer Science and Engineering\\IIT Kharagpur 721302 \\ India}
\begin{document}
\maketitle

\begin{abstract}
We introduce the notion of the { \it strong $(r,p)$ cover} number $\chi^c(G,k,r,p)$
for $k$-uniform hypergraphs $G(V,E)$, where
$\chi^c(G,k,r,p)$ denotes the minimum number of $r$-colorings of vertices in $V$
such that each hyperedge in $E$ contains at least $min(p,k)$ vertices of distinct
colors in at least one of the  
$\chi^c(G,k,r,p)$ 
$r$-colorings.
We derive the exact values of $\chi^c(K_n^k,k,r,p)$ for small 
values of $n$, $k$, $r$ and $p$, where $K_n^k$ denotes 
the complete $k$-uniform hypergraph of $n$ vertices.
We study the variation of $\chi^c(G,k,r,p)$ with respect to 
changes in $k$, $r$, $p$ and $n$; we show 
that $\chi^c(G,k,r,p)$ is at least 
(i) $\chi^c(G,k,r-1,p-1)$, and, 
(ii) $\chi^c(G',k-1,r,p-1)$, 
where $G'$ is any $(n-1)$-vertex induced sub-hypergraph of $G$.
We 
establish a general
upper bound for $\chi^c(K_n^k,k,r,p)$ 
for complete $k$-uniform hypergraphs using a divide-and-conquer 
strategy for arbitrary values of $k$, $r$ and $p$.
We also 
relate $\chi^c(G,k,r,p)$ to the number $|E|$ of hyperedges, and the 
maximum {\it hyperedge degree (dependency)} $d(G)$, as follows. 
We show that 
$\chi^c(G,k,r,p)\leq x$ for integer $x>0$, if 
$|E|\leq \frac{1}{2}({\frac{r^k}{(t-1)^k \binom{r}{t-1}}})^x $, for any $k$-uniform hypergraph. 
We prove that
a { \it strong $(r,p)$ cover} of size $x$ can be computed in
randomized polynomial time
if 
$d(G)\leq \frac{1}{e}({\frac{r^k}{(p-1)^k \binom{r}{p-1}}})^x-1$. 
\end{abstract}

\bigskip


\section{Introduction}
\label{sec:intro}

Let $G(V,E)$ be an $n$-vertex $k$-uniform hypergraph. 
A hyperedge $e \in E$ is said to be {\it properly $(r,p)$  colored} 
by a $r$-coloring of vertices in $V$
if at least $\min(p,|e|)$ vertices in the hyperedge $e$ are 
distinctly colored by the $r$-coloring.
A {\it proper $(r,p)$  coloring} of $G$ is an 
$r$-coloring of the vertices in $V$, 
such that each hyperedge $e \in E$ is properly $(r,p)$ colored.
For some values of $r$ and $p$, $G$ may not have any 
proper $(r,p)$ coloring. 
The decision problem of 
determining whether $G$ has a proper $(r,p)$ coloring is NP-complete because
(i) the decision problem of bicolorability of hypergraphs is NP-complete 
\cite{lovasz1973coverings}, and (ii) proper $(2,2)$ coloring of $G$
is the same as proper bicoloring of $G$. 
Therefore, we consider the following problem of proper $(r,p)$ coloring of a hypergraph.
We seek to find a set of $r$-colorings $C=\{X_1,X_2...,\}$ 
such that there are 
at 
least $\min(p,|e|)$ distinctly colored vertices in  
every hyperedge $e\in E$
in at least one $r$-coloring in $C$.
Such a set $C$ is called a {\it strong $(r,p)$ cover} for $G$.
We define the {\it strong $(r,3)$ cover} 
number $\chi^c(G,k,r,3)$ as the minimum number of $r$-colorings 
required such that each hyperedge contains at least three 
vertices of different colors in at least one of the 
$\chi^c(G,k,r,3)$
$r$-colorings. In general, 
the {\it strong $(r,p)$ cover} number $\chi^c(G,k,r,p)$ is 
defined as the minimum number of $r$-colorings of vertices 
required such that each hyperedge contains at least 
$\min(p,|e|)$ vertices of different colors in at least one of the
$\chi^c(G,k,r,p)$
$r$-colorings. 

Our problem of finding a strong $(r,p)$ cover 
for hypergraphs is motivated by
some important concepts in combinatorics and graph theory: 
bicoloring covers \cite{tmspp2014}, graph decomposition \cite{ChuGra1981}, strong coloring of hypergraphs \cite{Berge1989,BWY2012},
and  can be viewed as a generalization of the notion of separating families \cite{renyi1961,kat1966,Kat1973,weg1979,dick1969,spencer1970,mao1983}.

\begin{definition}\cite{tmspp2014}
A set of bicolorings $C=\{X_1,...,X_t\}$ is 
called a bicoloring cover for a $k$-uniform hypergraph $G(V,E)$
if every hyperedge $e \in E$ is 
properly colored in at least one bicoloring $X_i$, $X_i \in C$.
The minimum cardinality of any such set 
$C$ is called the bicoloring cover number $\chi^c(G)$.
\end{definition}
In \cite{tmspp2014}, we have shown that $\chi^c(K_n^k)$ is equal to $\left\lceil \log{\frac n {k-1}} \right\rceil$; and $\chi^c(G) = \lceil\log \chi(G) \rceil$ for arbitrary $k$-uniform hypergraphs.
Algorithms are also presented for computing bicoloring covers of size $\log |M|+2$ and $\lceil \log  \frac{|H|}{k-1}\rceil + 1$, where $M$ and $H$ denote a matching and  a hitting set, respectively.
We have also shown that a bicoloring covers of size $x$ for a $k$-uniform hypergraph $G$
can be obtained if (i) the number of hyperedges $m$ is less than or equal to $2^{(k-1)x-1}$ , or
(ii) the dependency $d$ is upper bounded by $\frac{2^{x(k-1)}}{e}-1$.

Observe that a bicoloring cover ($r$-coloring cover) of 
$G$ is equivalent to a strong $(2,2)$ cover (strong $(r,2)$ cover) 
for 
$G$. 
So, a strong $(r,p)$ cover  is a
generalization of a $r$-coloring cover, where, 
instead of two vertices of distinct colors, 
at least $p$ vertices of distinct colors are 
required for proper coloring of hyperedges.
Note that $\chi^c(G,k,r,2)=\chi^{c}_r(G)$ 
($r$-coloring cover number of $G$ \cite{tmspp2014}) and
$\chi^c(G,k,r,3)\geq \chi^c_r(G)$, $r\geq 3$.

\begin{definition}
Decomposition of a graph $G(V,E)$ deals with finding a family of graphs ${\cal{H}}=\{H_1,...,H_j\}$ such that
(i) $V(H_i) \subseteq V(G)$ for all ${H_i \in {\cal{H}}}$,
(ii) $\cap_{H_i \in \cal{H}} E(H_i)=\phi$, and
(iii) $\cup_{H_i \in \cal{H}} E(H_i)=E(G)$.
\end{definition}

The family ${\cal{H}}$ may consist of paths, cycles, bipartite graphs or matchings and there is vast literature for various kinds of decomposition of graphs 
(see \cite{ChuGra1981}).

A Strong coloring (see \cite{Berge1989}) of a $k$-uniform hypergraph $H(V,E)$ is 
the problem of $r$-coloring the vertices in $V$ in such a way 
that for every hyperedge $e \in E$, each vertex in $e$ gets 
a different color, for some $r \in {\mathbf{N}}$.
In other words, it is a function $C: V \rightarrow {1,...,r}$ such that for each hyperedge $e \in E$ 
if $v_1 \in e$ and $v_2 \in e$, then $C(v_1) \neq C(v_2)$.
The minimum value of $r$ for which $H$ admits a strong coloring 
is called the strong chromatic number of $H$, denoted by $\gamma(H)$.
Blais et.al. \cite{BWY2012} proposed a notion of semi-strong colorings of hypergraphs.

\begin{definition}[\cite{BWY2012}]
For a fixed $p \geq 2$, a $p$-strong coloring of $H$ is an assignment of colors to the vertices in $V$ such that every hyperedge $e \in E$ gets $\min\{p,|e|\}$ distinctly colored vertices. The $p$-strong chromatic number of H, denoted $\chi(H,p)$, is the minimum number of colors required to $p$-strong color $H$.
\end{definition}

The 2-section of a hypergraph \cite[p. 37]{Berge1989} $H(V,E)$ is the graph $[H]_2(V,E')$, where $(v_i,v_j) \in E'$ if $v_i \in e$ and $v_j \in e$, for some $e \in E$, $v_i,v_j \in V$. It is  easy to see that $\chi(H,k)=\chi([H]_2)$, where $\chi([H]_2)$ denote the chromatic number of $[H]_2$. So, the general results from chromatic numbers of graphs can be used for strong coloring of hypergraphs as well.
However, observe that a strong $(r,r)$ cover of $G$ cannot be directly 
compared with a $r$-coloring cover 
(see \cite{tmspp2014}) of the 2-section graph $[G]_2$ since 
we cannot directly enforce that the clique in $[G]_2$ corresponding 
to a hyperedge in $G$ is covered by a single $r$-coloring.

A strong $(r,p)$ cover can be viewed as as a 
minimal separator for the hyperedges of a 
hypergraph under certain restrictions.
A strong $(r,p)$ cover of 
$G$ is a set $C$ of $r$-colorings, $C=\{X_1,X_2...,\}$, 
such that each hyperedge $e\in E$ consists of at least $p$ 
distinctly colored vertices in at least one $r$=coloring
$X_i$, $X_i \in C$. 
Observe that the 
hyperedges properly $(r,p)$ colored by 
any $r$-coloring $X_i$ are hyperedges of an $(r,p)$-colorable 
$k$-uniform sub-hypergraph 
$H_i$  of $G$, where each hyperedge in $H_i$ consists of at least 
$p$ distinctly colored vertices. 
So the family of hypergraphs ${\cal{H}}=\{H_1,H_2...,H_{|C|}\}$, may 
be considered to be a decomposition of $G$, 
where $H_i \cap H_j$ may or may not be empty.

A strong $(r,p)$ cover can be viewed as a problem of minimal separator for hyperedges of a hypergraph under restrictions.
\begin{definition}
Let $[n]$ denote the set ${1,2,...,n}$. A set $S \subseteq [n]$ separates $i$ from $j$ if $i \in S$ and
$j \not\in S$. A set ${\cal S}$  of subsets of $[n]$ is a separator if, for each $i,j \in [n]$ with $i\neq j$,
there is a set $S$ in ${\cal S}$ that separates $i$ from $j$. If, for each $(i,j) \in  [n]*[n]$ with $i\neq j$,
there is a set $S\in{\cal S}$ that separates $i$ from $j$ and a set $T \in {\cal S}$ that separates $j$ from $i$, then ${\cal S}$ is called a complete separator. 
Moreover, with the additional constraint that the sets $S$ and $T$ that separate $i,j$ are required to be disjoint, then ${\cal S}$ is called a total separator.
\end{definition}
Let $f_0(n)$, $f_1(n)$, $f_2(n)$ denote the size of a minimal separator, a minimal complete separator and a minimal total separator of $[n]$, respectively.
R\'{e}nyi \cite{renyi1961} introduced the notion of separators and 
showed that $f_0(n)= \lceil  \log_2 n\rceil$. Let $f_0(n, _\leq k)$ and $f_0(n,_=k)$ denote the size of a smallest separator where each set is constrained to have at most $k$ elements and exactly $k$
elements, respectively.
Katona \cite{kat1966,Kat1973} showed that these two quantities are identical and
established (i) $f_0(n, _\leq k) = \lceil \log_2 n  \rceil$, if $k \geq \lfloor \frac{n}{2}\rfloor$, and
(ii)$\frac{n \log_2 n}{k \log_2 (\frac{en}{k})} \leq f_0(n, _\leq k) \leq )\frac{n \log_2 (2n)}{k \log_2 (\frac{n}{k})} $, if $k < \lfloor \frac{n}{2}\rfloor$. Wegener \cite{weg1979} improved the upper
bound to $\lceil \frac{\log_2 n}{\log_2 \lceil\frac{n}{k}\rceil}\rceil (\lceil\frac{n}{k}\rceil-1)$ for $k < \lfloor \frac{n}{2}\rfloor$.
Dickson \cite{dick1969} introduced the problem of complete separators and showed that $\lim \frac{f1(n)}{\log n} =1$.
Spencer \cite{spencer1970} showed that $f1(n)$ is the minimum $t$ such that $\binom{t}{\frac{t}{2}} \geq n$ using Sperner's Lemma.
He gave a simpler formula for $f1(n)$ as 
$f1(n)= \log_2 n + \log_2 \log_2 n + \frac{1}{2} \log_2 (\frac{\pi}{2})+o(1)$.
The notion of total separating system was introduced by Katona \cite{Kat1973}.
Mao-cheng \cite{mao1983} showed that $\min\{k+f2(\lceil \frac{n}{k} \rceil)|k=2,...,n-2\} \leq f2(n) \leq \min\{ f2(k)+f2(\lceil \frac{n}{k} \rceil)|k=2,...,n-2\}$, which gives 
$f2(n)=\min\{2p+3\lceil \log_3 (\frac{n}{2^p}) \rceil|p=0,1,2\} $.
\begin{definition}
Given a graph $G=(V,E)$, a family ${\cal S}=\{S_1,...,S_t\}$ of subsets of $V$ is called an
$E$-separating family if for any edge $e=\{x,y\} \in E$, there are disjoint $S_i$
and $S_j$ such that $x \in S_i$ and $y \in S_j$.
\end{definition}
For $G=(V,E)$, let $g(G)$ denote the minimum number of subsets in an $E$-separating family.
Edmond's problem (see \cite{mao1983}) is to determine $g(G)$ for a graph $G=(V,E)$.
Mao-cheng \cite{mao1983} demonstrated that $g(G)=f2(\chi(G))$.

We pose the following separation problem for $k$-uniform hypergraphs.
\begin{definition}
For a $k$-uniform hypergraph  $G(V,E)$, $V=[n]$,
let ${\cal S}=\{S_1,...,S_t\}$, where $\forall S_i \in {\cal S}$, $S_i=\{S_{i1},...,S_{ir}\}$, such that
(i) $S_{1j}\subset V$, $S_{1j} \neq \phi$, $1 \leq j \leq r$, 
(ii) $\cup_{j}S_{ij}=V$, and
(iii) $\forall e \in E$, there exists at least one $S_i$ such that $e$ has non-empty intersection with at least $p$ elements of $S_i$.
Then, ${\cal S}$ is called a $(r,p)$ separator for the hyperedges of $G$.
\end{definition}
It is not hard to see that the minimum cardinality of ${\cal S}$ is equal to the the strong $(r,p)$ cover number 
$\chi^c(G,k,r,p)$ for $G$.
 

Consider the scheduling of a movie carnival. 
Suppose there are $n$ movies to be played in the  carnival. 
Let $m$ be the number of viewers. 
Each viewer has a set of $k$ movies that he/she wishes to watch. 
There are $r$ time slots for playing movies per day. Each movie can be played only once during a single day. A viewer cannot watch two movies in the same time slot. In order to promote the event, the organizers  decide to give away some prizes to the set of viewers who have watched at least $p$ movies in a single day ($r  \geq p$). Given $n$ movies, $m$ viewers, and $r$ time slots, is it possible to schedule the movies in a single day in such a way that every viewer can be eligible for the prize i.e. every viewer has a chance to watch $p$ out of his list of $k$ movies?
This problem can be modeled by a strong $(r,p)$ coloring of the $k$-uniform hypergraph $G(V,E)$, where 
\begin{itemize}
\item $V$ is the set of $n$ vertices $\{V_1,...,V_n\}$, where each $V_i$ denote a distinct movie, for $1 \leq i \leq n$;
\item  $E$ is the set of $m$ hyperedges $\{E_1,...,E_m\}$, where each $E_j \subset V$ denote the list of $k$ movies in the wish list of $jth$ viewer, for $1 \leq j \leq m$.
\end{itemize}
However, there are hypergraphs which certainly do not have a strong $(r,p)$ coloring. So it may not be possible to give every viewer an opportunity to be selected for the prize.
In order to resolve this problem, the organizers decide to run the carnival for multiple days. This is equivalent to a set of strong $(r,p)$ colorings.
Now the problem is to minimize the number of days : ``Given $n$ movies, 
$m$ viewers, and $r$ time slots, what is the minimum number of days 
the carnival must run so that every viewer gets a chance to be 
eligible for the prize (i.e for every viewer, there is at least one day 
such that he/she can watch at least $p$ movies out of his list of $k$ movies)''.
Observe that the minimum number of days for which the carnival must run is the 
strong $(r,p)$ cover number $\chi^c(G,k,r,p)$ of $G$.

Unless otherwise stated, $G$ denotes a $k$-uniform hypergraph, 
having vertex set $V$ (or $V(G)$) and hyperedge set $E$ (or $E(G)$).
Observe that $\chi(G,k,r,k)=\chi(G,k,r,p)$ for all values of 
$p \geq k$. We assume $k$ is at least
$p$ in the rest of the paper.

%


\subsection{Organization of the paper}

In Section \ref{sec:preliminaries}, we describe some terminologies, preliminary results on the maximum number of hyperedges that can be covered by any $r$-coloring and prove some exact results and bounds for $\chi^c(K_n^k,k,r,p)$ for small values of $n$, $k$, $r$ and $p$. In particular, we show that $\chi^c(K_7^3,3,3,3)=\chi^c(K_8^3,3,3,3)=\chi^c(K_9^3,3,3,3)=4$ and we investigate for every valid combination of $k$, $r$ and $p$ for $n \leq 7$.

In Section \ref{sec:lowerbounds}, we show that $\chi^c(G,k,r,p)$ 
is at least (i) $\chi^c(G,k,r,p-1)$, 
(ii) $\chi^c(G,k,r-1,p-1)$, and
(iii) $\chi^c(G',k-1,r,p-1)$, where $G'$ is any 
$(n-1)$ vertex sub-hypergraph of $G$.

In Section \ref{sec:upperbounds}, we establish a upper bound on $\chi^c(G,k,3,3)$ of $\frac{1}{3}(\frac{n}{k-1})^{2.71} + \log_{\frac{3}{2}}\frac{n}{k-1}-1$ if $k$ is odd and $\frac{1}{3}(\frac{n}{k-2})^{2.71} + \log_{\frac{3}{2}}\frac{n}{k-2}-1$ if $k$ is even.
We show that the upper bound reduces drastically (to $O(n^{1.59})$ from $O(n^{2.71})$ for a fixed $k$) when  4 colors instead of 3 colors are used.
We prove the general upper bound of $\chi^c(G,k,r,p)\leq \chi^c(K_n^k,k,r,p) \leq \Big(\frac{n(p-1)}{(k-1-l)r}\Big)^{\log_{\frac{r}{p-1}}D}+\log_{\frac{r}{p-1}}{\Big(\frac{n}{k-1-l}\Big)}-1$,  for some fixed $l$, $0 \leq l \leq p-2$, on $\chi^c(K_n^k,k,r,p)$, where $D=\lceil \frac{\binom{r}{p-1}}{\lfloor \frac{r}{p-1}\rfloor}\rceil.$

In Section \ref{sec:semichromatic}, we relate the $p$-strong chromatic number $\chi(G,p)$ and strong $(r,p)$ cover number $\chi^c(G,k,r,p)$ and show that ${\chi^c(G,k,r,p)} \geq \log_r \chi(G,p)$.

In Section \ref{sec:probabilistic}, we show that every $k$-uniform hypergraph $G(V,E)$ with number of hyperedges $|E|$ less than or equal to $\frac{1}{2}({\frac{r^k}{(p-1)^k \binom{r}{p-1}}})^x$ has a strong $(r,p)$ cover of size $x$. We also prove that if the maximum dependency of any hyperedge of $G$ is  less than or equal to $\frac{1}{e}({\frac{r^k}{(p-1)^k \binom{r}{p-1}}})^x-1$, then $\chi^c(G,k,r,p) \leq x$.

\section{Preliminaries}
\label{sec:preliminaries}

Given fixed $r$ and $p$, a hypergraph with too many hyperedges may 
have
no $r$-coloring of its vertices
enforcing a proper $(r,p)$ coloring on every hyperedge of the hypergraph. 
It is interesting to find the maximum number $m$ of hyperedges such 
that (i) there exists a $k$-uniform hypergraph $G(V,E)$ with 
a proper $(r,p)$ coloring, and (ii) 
$G(V,E\cup e)$ does not have any proper $(r,p)$ coloring  
for every hyperedge 
$e \not\in E$.
Let $M(n,k,r,p)$ denote the maximum number of hyperedges of 
any $n$-vertex $k$-uniform hypergraph $G(V,E)$ that can be 
properly $(r,p)$ colored by a single $r$-coloring.
We have the following theorem from \cite{tmspp22015}.

\begin{theorem}\label{thm:boundminmax}\cite{tmspp22015}
For a fixed $n$, $k$, $r$ and $p$, the maximum number of properly $(r,p)$ colored $k$-uniform hyperedges $M(n,k,r,p)$ on any $n$-vertex hypergraph $G$ is at most  $\binom{n}{k} - \sum_{i=1}^{p-1} m(n_1,k,r,i)$ and at least 
$\binom{n}{k} - \sum_{i=1}^{p-1} m(n_2,k,r,i)$,
where  $n_1=\lfloor\frac{n}{r} \rfloor \cdot r$, $n_2=\lceil \frac{n}{r} \rceil \cdot r$, and
$m(n',k,r,i)= \binom{r}{i}\Big( \binom{\frac{n'}{r} i}{k}- i\binom{ \frac{n'}{r} (i-1)}{k}+\binom{i}{2}\binom{ \frac{n'}{r} (i-2)}{k}... \allowbreak (-1)^c \binom{i}{c}\binom{ \frac{n'}{r}  c}{k} \Big)$, and,
$c$ is the smallest integer such that $\frac{n'}{r} c >= k$.
Moreover,  the number of properly $(r,p)$ colored hyperedges is maximized when the $r$-coloring is balanced.
\end{theorem}

\begin{corollary}\label{cor:exactmax}\cite{tmspp22015}
The maximum number of properly $(r,r)$ colored hyperedges of a $K_n^r$ in any $r$-coloring
(i) is $M(n,r,r,r)=|A_1||A_2|...|A_r|$, and,
(ii) the $r$-coloring that maximizes the number of properly colored hyperedges splits the vertex set into equal sized parts.

\end{corollary}

Theorem \ref{thm:boundminmax} directly implies the following lower bound for $\chi^c(G,k,r,p)$ based on the number of
hyperedges of $G$.

\begin{corollary}\label{cor:almostequal}

Let $G(V,E)$ be any $n$-vertex 
$k$-uniform hypergraph, and $r$ and $p$ be fixed positive integers, where
$r \geq p$. 
Then, 
then $\chi^c(G,k,r,p) \geq \frac {|E|}  {M(n,k,r,p)}$. 

\end{corollary}

Using Theorem \ref{thm:boundminmax}, Corollary \ref{cor:exactmax}, and Corollary \ref{cor:almostequal},
we have enumerated all the exact values of $\chi^c(K_n^k,k,r,p)$ for all valid combinations of the parameters upto $n \leq 7$ (see Table \ref{tab:values}).
We demonstrate the proof for exact values of $\chi^c(K_n^k,k,r,p)$ only for $n=7,8,9$, and $k=r=p=3$.

\begin{theorem}
$\chi^c(K_7^3,3,3,3)=\chi^c(K_8^3,3,3,3)=\chi^c(K_9^3,3,3,3)=4$.
\end{theorem}
\begin{figure}
	\centering
	\includegraphics[scale=0.7]{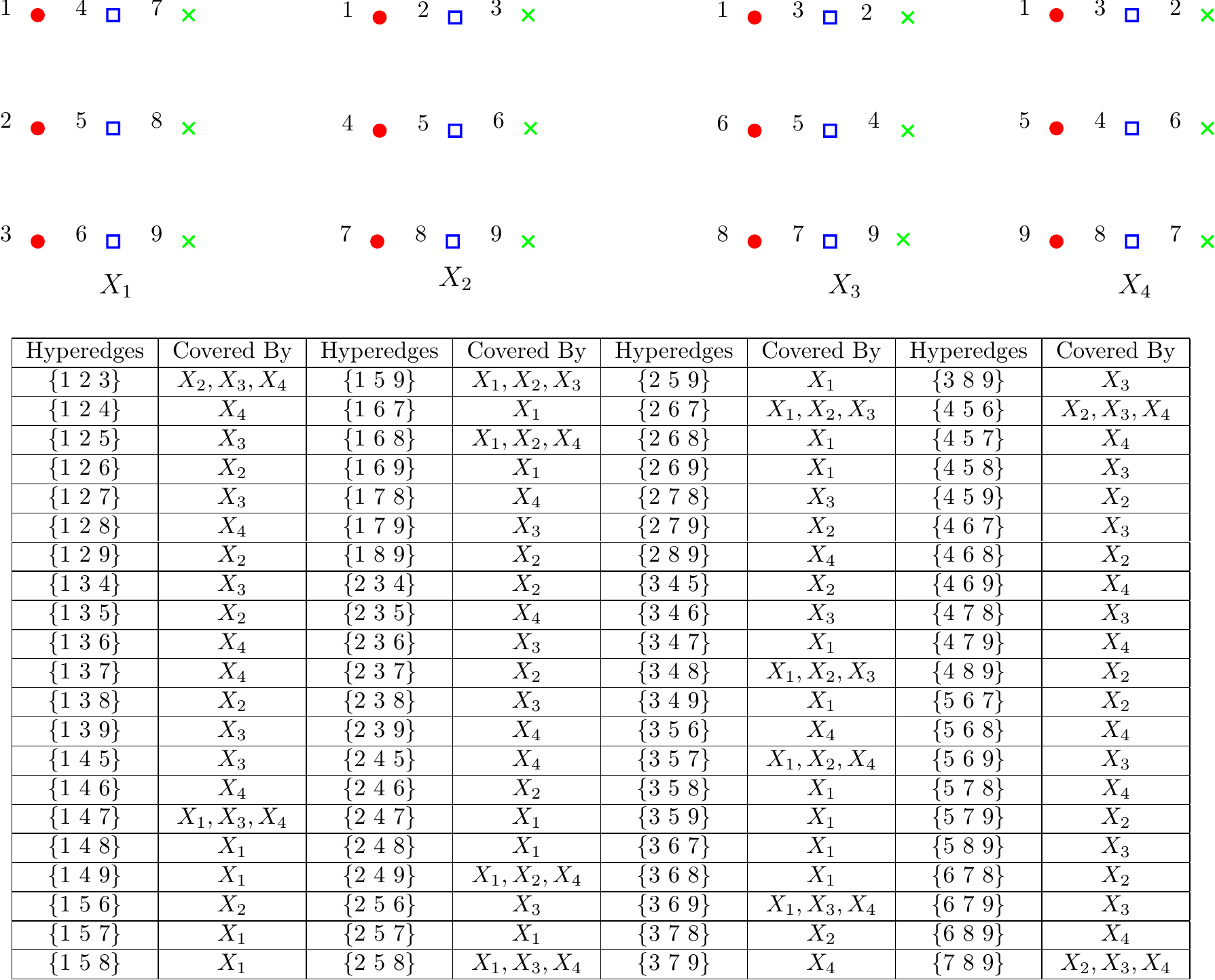}
	\caption{$\chi^c(K_9^3,3,3,3)=4$: $C_{11}=\{X_1,X_2,X_3,X_4\}$ is a proper (3,3) cover for $K_9^3$.}
	\label{fig:k9333}
\end{figure}
\begin{proof}
We show the relations by proving (1). $\allowbreak\chi^c(K_7^3, 3,3,3) \geq 4$, and (2). $\chi^c(K_9^3,3,3,3)\leq4$.
$K_7^3$ contains 35 hyperedges. Using Corollary \ref{cor:exactmax}, the maximum number of hyperedges that can be properly (3,3) colored is 12 hyperedges using an almost equal partition. Let these colorings be Type I coloring.
We consider two mutually exclusive cases: colorings that include at least one Type I coloring and colorings that do not contain any Type I coloring.
\begin{enumerate}[leftmargin=*]
\item Without loss of generality, let the coloring be $X_1=\{\text{red: \{$v_1,v_2,v_3$\}, blue: \{$v_4,v_5$\}, green: \{$v_6,v_7$\}}\}$: a Type I coloring. 
	We can verify with a simple code that after $X_1$, any coloring can properly (3,3) color at most 9 new hyperedges.
	If we do not use any coloring that divides the vertex set into almost equal partitions (i.e. type-1 coloring), then it is easy to see that the maximum number of hyperedges that can be properly (3,3) colored cannot exceed 8.
	we need at least 5 type-2 colorings to cover the hypergraph. 
	So any set of three 3-colorings can properly (3,3) color at most $12+9+9=30<35$ hyperedges. So we need at least 4 3-colorings of vertices to cover $K_7^3$ and hence $\chi^c(K_7^3,3,3,3) \geq 4$.
	\item  Let $X_1=\{\text{ red: \{$v_1,v_2,v_3$\}, blue: \{$v_4,v_5,v_6$\}, green: \{$v_7,v_8,v_9$\}}\}$, 
	$X_2=\{\text{ red: }\{v_1,v_4,v_7\},\text{ blue: }\{v_2, \allowbreak v_5,  v_8\},\text{ green: }\{v_3, v_6,v_9\}\}$, $X_3=\{\text{red:\{$v_1,v_6,v_8$\}, blue: \{$v_3,v_5,v_7$\}, green: \{$v_2,v_4,v_9$\}}\}$, 
	$X_4=\{\text{red: }\{\allowbreak v_1,v_5,v_9\},\text{ blue: }\{v_3,v_4,v_8\},\text{ green: }\{v_2,\allowbreak v_6,v_7\}\}$. 
Let $C_{11}=\{X_1,X_2,X_3,X_4\}$. From Fig. \ref{fig:k9333}, it is clear that $C_{11}$ properly (3,3) colors every hyperedge of $K_9^3$. So 
	$\chi^c(K_9^3,3,3,3) \leq 4$.
\end{enumerate}
\qed
\end{proof}
\begin{table}[!htb]
	\centering
	\scalebox{0.85}{
		\begin{tabular}{|c|c|c|c|c|c|c|c|c|c|c|c|c|c|c|c|}
			\hline 
			\backslashbox{$n$, $k$}{$r$, $p$} & 3,3 & 4,3 & 4,4 & 5,3 &5,4 &5,5 &6,3 &6,4 & 6,5 & 6,6 & 7,3&7,4 &7,5 & 7,6 & 7,7\\\hline
			4,3 & 2 & 1 & 1 & & & & & & & & & & &  & \\ \hline
			4,4 & 1 & 1 & 1 & & & & & & & & & & &  & \\ \hline
			5,3 & 3 & 2 & 2 & 1 & & & & & & & & & &  & \\ \hline
			5,4 & 2 & 1 & 3 & 1 & 1 & 1 & & & & & & & &  & \\ \hline
			5,5 & 1 & 1 & 1 & 1 & 1 & 1 & & & & & & & &  & \\ \hline
			6,3 & 3 & 3 & 3 & 2 & 2 & 2 & 1 & 1 & 1 & 1 & & & &  & \\ \hline
			6,4 & 2 & 2 & (4,5) & 1 & 3 & 3 & 1 & 1 & 1 & 1 & & & &  & \\ \hline
			6,5 & 1 & 1 & 2 & 1 & 1 & 3 & 1 & 1 & 1 & 1 & & & &  & \\ \hline
			6,6 & 1 & 1 & 1 & 1 & 1 & 1 & 1 & 1 & 1 & 1 & & & &  & \\ \hline
			7,3 & 4 & 3* & (6,-)& 2 & 2 & 2 & 2 & 2 & 2 & 2 & 1 & 1 & 1 & 1 & 1 \\ \hline
			7,4 & (3,-*) & 2 & (5,-) & 2 & 3 & 3 & 1 & 3 & 3 & 3 & 1 & 1 & 1 & 1 & 1 \\ \hline
			7,5 & 2 & 1 & (3,-*)& 2 & 2 & (6,-)& 1 & 1 & 3 & 3 & 1 & 1 & 1 & 1 & 1 \\ \hline
			7,6 & 2 & 1 & 2 & 1 & 1 & 2 & 1 & 1 & 1 & (4,-) & 1 & 1 & 1 & 1 & 1 \\ \hline
			7,7 & 1 & 1 & 1 & 1 & 1 & 1 & 1 & 1 & 1 & 1 & 1 & 1 & 1 & 1 & 1 \\ \hline
			8,3 & 4 &  & & & & & & & & & & & &  & \\ \hline
			9,3 & 4 &  & & & & & & & & & & & &  & \\ \hline
		\end{tabular}}
		\label{tab:values}
		\caption{$\chi^c(K_n^k,k,r,p)$ for small values of $n$, $k$, $r$ and $p$. Table value $(a,b)$ indicate that $a$ is the lower bound and $b$ is the upper bound for corresponding $\chi^c(K_n^k,k,r,p)$.
			For the '*' indicated values, the proof is given in the appendix.}
	\end{table}
	
	\section{Lower bounds for \texorpdfstring{$\chi^c(K_n,k,r,p)$}{}}
\label{sec:lowerbounds}

In this section, we study the behavior of the strong $(r,p)$ cover number with respect to changes in its parameters, namely $n$, $k$, $r$ and $p$. 

\begin{proposition}
\label{prop:lb}
Let $G_n$ denote a $n$ vertex $k$-uniform hypergraph. Then,
\begin{align*}
\chi^c(G_n,k,r,p) \geq \begin{cases}
	\chi^c(G_n,k,r+1,p). \hspace{2cm}\hfill (1)\\
    \chi^c(G_n,k,r,p-1).  \hfill (2)\\
    \chi^c(G_n,k,r-1,p-1). \hfill (3) \\
    \chi^c(G_{n-1},k-1,r,p-1). \hfill (4) 
\end{cases}
\end{align*}
\end{proposition}
\begin{proof}

(1) To see that $\chi^c(G_n,k,r,p) \geq \chi^c(G_n,k,r+1,p)$,
observe that any strong $(r,p)$ cover for $G_n$ is a trivial strong $(r+1,p)$ cover for $G_n$.

(2) The first inequality follows from the definition: if some hypergraph $G_n$ on $n$ vertices has a strong $(r,p)$ cover of size $\chi^c(G_n,k,r,p)$, the same cover ensures a $(r,p-1)$ strong cover itself and the cover number can only reduce with reduction of $p$. Therefore, $\chi^c(G_n,k,r,p)$ is an increasing function with respect to $p$.

(3) Let $\chi^c(G_n,k,r,p)=l$. Consider some set $C$ of $l$ $r$-colorings $\{X_1,...X_l\}$ that properly $(r,p)$ covers all the hyperedges of $G_n$,
where each $X_i$, $1 \leq i \leq l$, is a $r$-coloring of vertices of $G_n$.
Construct a new set $C'=\{X_1',...X_l'\}$ of $(r-1)$ colorings from $C$, $|C'|=|C|$, as follows: in $X_i'$, color classes of color 1 through $r-1$ remain the same as in $X_i$;  the vertices of $r^{th}$ color in $X_i$ are arbitrarily added  to any one of the other color class (let $(r-1)^{th}$ color class) in $X_i'$.
So, each $C'$ consists of a set of $(r-1)$ colorings of vertices of $G_n$.
Consider any hyperedge $e \in E(G)$ that is properly $(r,p)$ colored by $X_i \in C$.
If $e$ does not contain any vertex colored with $r$ in $X_i$, then $e$ still consists of $p$ distinctly colored vertices in $X_i'$, and so, it is properly $(r-1,p-1)$ colored.
If $e$ contains some vertex colored with $r$ in $X_i$, since it is properly $(r,p)$ colored, it must contain vertices of $p-1$ other colors, which ensures a proper $(r-1,p-1)$ coloring of $e$ in $X_i'$.

(4) Let the vertices of $G_n$ be $N=\{1,2,...,n\}$. 
Let $\chi^c(G_n,k,r,p)=l$. Consider some set $C$ of $l$ $r$-colorings $\{X_1,...X_l\}$ that properly $(r,p)$ covers all the hyperedges of $G_n$, where each $X_i$, $1 \leq i \leq l$, is a $r$-coloring of vertices of $G_n$.
Construct a new hypergraph $G_{n-1}$ on vertex set $N \setminus {n}$ and hyperedge set 
\begin{align}
E(G_{n-1}) = \begin{cases}
    \{e-{n}|e \in E(G)\} \text{if $n \in e$}.\\
    \text{all the $k-1$ sized sets of $e$ otherwise} .
  \end{cases}
\end{align}
Construct a new set $C'=\{X_1',...X_l'\}$ of $r$-colorings from $C$, $|C'|=|C|$, as follows: in $X_i'$, color every vertex $v_j \in \{1,2,...,n-1\}$ with the same color as the color of $v_j$ in $X_i$.
Every hyperedge $e' \in E(G_{n-1})$ belongs to one of the following two exhaustive categories, based on whether its corresponding hyperedge $e \in E(G_n)$ contains $n$.
(i) $e \in E(G_n)$ and $n \in e$: The corresponding hyperedge $e' = e \setminus n$ in $E(G_{n-1})$ consists of exactly $k-1$ vertices. Observe that $C'$ properly $(r,p-1)$ covers $e'$: if  $X_i$ properly  $(r,p)$ colors $e$ in $C$, then $X_i'$ properly  $(r,p)$ colors $e'$ in $C'$.
(ii) $e \in E(G_n)$ and $n \not\in e$: 
From definition of $E(G_{n-1})$, exactly
$k$ hyperedges of size $k-1$ derived from $e$ in  $E(G_{n-1})$.
Note that if $X_i \in C$ properly $(r,p)$ covered $e$, then  
 every hyperedge derived from $e$ in $E(G_{n-1})$
 can be properly $(r,p-1)$ covered by the $r$-coloring $X_i' \in C'$.
 
 This concludes the proof of all the statements of Proposition 
\ref{prop:lb}.
\qed
\end{proof}

\section{General upper bounds}
\label{sec:upperbounds}
In this section, we derive upper bounds on $\chi(G,k,r,p)$ based on the parameters.
Let $\chi_e(G)$ denote the chromatic index of hypergraph $G$, i.e., the minimum number of colors required to color the hyperedges such that no two intersecting hyperedges receive the same color.
Note that each color class of an optimal hyperedge coloring of $G$ splits the hypergraph into $\chi_e(G)$ disjoint matchings.
Observe that every hyperedge of each of the matchings can be properly $(r,p)$ colored by a single $r$-coloring.
So, $\chi_e(G)$ $r$-colorings are always sufficient for a strong $(r,p)$ cover of any hypergraph $G$.
\begin{observation}
For any $k$-uniform hypergraph $G$, $\chi(G,k,r,p) \leq \chi_e(G)$.
\end{observation}
Vizing's theorem (\cite{vizing1964},\cite[pp. 277-278]{West2000}) states that every simple graph $H$ has an edge coloring using $\Delta(H)+1$ colors, where $\Delta(H)$ is the maximum degree of any vertex in $H$.
For simple hypergraphs (i.e. two hyperedges can share at most 1 vertex), Kahn \cite{Kahn199231} showed that  $\chi_e(G)$ is upper bounded by $n+o(n)$.
Erdos et.al. \cite{Erdos198125} also conjectured that $\chi_e(G)$ is $n$ for simple hypergraphs.
A corollary of  Baranayi's theorem \cite{Baranyai1979276,Berge197765} states stat
chromatic index of a complete $n$-vertex $k$-uniform hypergraph is $\lceil \frac{\binom{n}{k}}{\lfloor \frac{n}{k}\rfloor}\rceil$.
In what follows we relate $\chi(G,k,r,p)$ to its parameters and compute strong $(r,p)$ covers.

Let $T(n)$ denote the number of $r$-colorings used by the algorithms in this section to properly $(r,p)$ cover all the hyperedges of $K_n^k$ (each hyperedge contains at least $p$ distinctly colored vertices in at least one of the $r$-colorings).
In particular, if  $r=3$ ($r=4$) and $p=3$ ($p=4$), $T(n)$ denotes the number of 3-colorings(4-colorings) used by the algorithm to properly $(3,3)$ ($(4,4)$) cover all the hyperedges of $K_n^k$.
We start the analysis of $\chi^c(G,k,r,3)$  for general hypergraphs with a complete $k$-uniform hypergraph ${K_n}^k$ with $r=3$, $k\geq 3$.

Consider the following algorithm $Knk33cover$ for computing a strong $(3,3)$ cover $C1$ for a ${K_n}^k$.
$Knk33cover$ 3-colors the vertices of a ${K_n}^k$ with each color class of almost equal sizes: the 3-coloring of vertices splits the vertex set into three color classes, $V_1$, $V_2$ and $V_3$, each of size either $\lfloor \frac{n}{3}\rfloor$
or $\lceil \frac{n}{3} \rceil$. 
Let this coloring be $X_1$.
The hyperedges containing at least one vertex from each of the three color classes is properly (3,3) colored by $X_1$. 
The hyperedges that are not properly (3,3) colored by the first 3-coloring $X_1$ consists of vertices of at most two of the three classes. So, there are 3 possibilities: the remaining hyperedges consists of vertices of only $V_1 \cup V_2$ or $V_2 \cup V_3$ or $V_3 \cup V_1$. 
Therefore, the problem of computing a strong $(3,3)$ cover for $K_n^k$ reduces to computing a strong $(3,3)$ cover for $K_{\lceil\frac{2n}{3}\rceil}^k$, i.e. the problem size reduces to $\lceil \frac{2n}{3}\rceil$ after the first 3-coloring. 

Now, we establish the base case for $Knk33cover$. 
Consider the case when $k$ is odd and $n \leq \frac{3(k-1)}{2}$. We can 3-color the vertices such that each color class is of size at most $\frac{(k-1)}{2}$. Merging any two color classes would result in at most $k-1$ vertices, so any hyperedge must contain at least one vertex each from the three color classes i.e.  every hyperedge is properly (3,3) colored. When $k$ is even, in a single 3-coloring, we can allow a color class of size $\frac{k}{2}$ and two color classes of size $\frac{k-2}{2}$ such that merging any two color classes would result in at most $k-1$ vertices, i.e. if $k$ is even and $n \leq \frac{3(k-1)-1}{2}$, a single 3-coloring can properly (3,3) color every hyperedge.

So, the following recurrence describes the upper bound for $\chi^c(K_n^k,k,3,3)$.
\begin{align*}
T(n) \leq \begin{cases}
    1,\text{ if}\begin{cases} n \leq \frac{3(k-1)}{2} \text{ and $k$ odd,}\\
    				 n \leq \frac{3(k-1)-1}{2} \text{ and $k$ even.}
    	\end{cases}\\
    3T(\lceil\frac{2n}{3}\rceil)+1,\text{ if}\begin{cases} n > \frac{3(k-1)}{2} \text{ and $k$ odd,}\\
    											 n > \frac{3(k-1)-1}{2}\text{ and $k$ even}.
    								\end{cases}
    \end{cases}
\end{align*}
 Solving the recurrence recursively for $i$ iterations (ignoring the ceiling), we get,
\begin{align*}
T(n) \leq 3^iT(\frac{n}{(\frac{3}{2})^i})+i.
\end{align*}
Solving the recurrences for the odd and even values of $k$ separately with proper substitutions,
we have the following theorem.
\begin{theorem}\label{thm:up33}
For any arbitrary $n$ vertex $k$-uniform hypergraph $G$,\\
$\chi^c(G,k,3,3) \leq \chi^c(K_n^k,k,3,3) \leq 
\begin{cases}
\frac{1}{3}(\frac{n}{k-1})^{2.71} + \log_{\frac{3}{2}}\frac{n}{k-1}-1,\text{ if $k$ is odd}, \\
\frac{1}{3}(\frac{n}{k-2})^{2.71} + \log_{\frac{3}{2}}\frac{n}{k-2}-1 ,\text{ if $k$ is even}.
 \end{cases}$
\end{theorem}
Observe that there is a large variation in upper bounds of $\chi^c(K_n^k,k,3,3)$ and $\chi^{c}_3(K_n^k)$ \cite{tmspp2014}.
Using Proposition \ref{prop:lb} inequality  (1), it is clear that increasing the number of colors in each of the colorings is a possibility for reducing the strong cover size.
So, we increase the number of colors used in each coloring to four colors and evaluate $\chi^c(K_n^k,k,4,3)$ i.e. minimum number of 4-colorings of vertices $V(K_n^k)$ such that every hyperedge in $K_n^k$ is properly (4,3) colored. 

Consider the algorithm $Knk43cover$ for computing a properly (4,3)  cover for $K_n^k$.
$Knk43cover$ performs a balanced 4-coloring of vertices into four color classes, $V_1,V_2,V_3,V_4$, each color class containing at most $\lceil\frac{n}{4}\rceil$ vertices. The hyperedges that comprise at least one vertex from any three color classes is  properly (4,3)  covered.  So, the hyperedges still uncovered consists of vertices of at most two of the four classes. 
It may seem that we need $\binom{4}{2}=6$ four colorings in the recursive step: the total number of 4-colorings computed by the algorithm to properly (4,3) cover every hyperedge of $K_n^k$ is given by the recurrence
$T(n) \leq 6T(\lceil\frac{n}{2}\rceil)+1$. 
However, we can use parallelism in coloring using the fact that the hyperedges  $e1 \subseteq V_1 \cup V_2$ and $e2 \subseteq V_3 \cup V_4$ are independent: these hyperedges can be covered simultaneously by the same four coloring. Similarly, hyperedges  $e3  \subseteq V_1 \cup V_3$ can be covered alongside hyperedges $e4 \subseteq V_2 \cup V_4$, and hyperedges  $e5  \subseteq V_1 \cup V_4$ can be covered alongside hyperedges $e6  \subseteq V_2 \cup V_3$. So, we need three 4-colorings in the recursive step. Proceeding in  the same direction as in the proof of Theorem \ref{thm:up33}, we get the following theorem.

\begin{theorem}
For any arbitrary $n$ vertex $k$-uniform hypergraph $G$,\\
$\chi^c(G,k,4,3) \leq \chi^c(K_n^k,k,4,3) \leq 
\begin{cases}
\frac{1}{3}(\frac{n}{k-1})^{1.59} + \log_{2}\frac{n}{k-1}-1,\text{ if $k$ is odd}, \\
\frac{1}{3}(\frac{n}{k-2})^{1.59} + \log_{2}\frac{n}{k-2}-1 ,\text{ if $k$ is even}.
 \end{cases}$
\end{theorem}
There is a reduction in the upper bound of $\chi^c(G,k,r,3)$ from $O(n^{2.71})$ to $O(n^{1.59})$ (for fixed constant $k \geq 3$) when we use four instead of three colors. This reduction is due to the parallelism achieved in the colorings  in the recursive step. This leads us to the conjecture that increasing the number of color in coloring  reduces the value of $\chi^c(G,k,r,3)$. In what follows, we investigate for the general case of $\chi^c(K_n^k,k,r,p)$ and  maximize this parallelism.

Consider a balanced $r$-coloring $X_1$ of vertices of $K_n^k$: this results in color partition ${\cal V}=\{V_1,...,V_r\}$ of the vertices, each of size $\lceil\frac{n}{r}\rceil$ (except possibly the last partition). Any hyperedge that includes at least one vertex from at least $p$ parts of ${\cal V}$ is properly $(r,p)$ colored by $X_1$. The hyperedges which are not properly $(r,p)$ colored by $X_1$ are completely contained inside any of the $(p-1)$ parts.
Let $\binom{[{\cal V}]}{p-1}$ denote the set of all the $(p-1)$ sized subsets of 
${\cal V}=\{V_1,...,V_r\}$ i.e. $\binom{[{\cal V}]}{p-1} = \{ \{V_1,V_2,...,V_{p-1} \}, \{ V_1,V_2,...,V_{p-2}, V_{p} \}, ...,\allowbreak \{ V_1,V_2,...,V_{p-2}, V_{r}\}, ..., \{ V_{r-p+2},V_{r-p+3},...,V_{r-1}, V_{r} \} \}$. $|\binom{[{\cal V}]}{p-1}|=\binom{r}{p-1}$.
So, a hyperedge $e$ is not properly $(r,p)$ colored by $X_1$ 
if there exists an element of $\binom{[{\cal V}]}{p-1}$ $\{V_i,...,V_j\}$ such that $e \subseteq V_i \cup ... \cup V_j$. 

We now analyze the hyperedges which can be covered simultaneously in subsequent $r$-colorings. Observe that 
two hyperedges $e1 \subseteq V_i \cup ... \cup V_j $ and $e2 \subseteq V_s \cup ... \cup V_t$, $\{V_i,...V_j\} \in \binom{[{\cal V}]}{p-1}$ and  $\{V_s,...V_t\} \in \binom{[{\cal V}]}{p-1}$, can be  properly $(r,p)$ colored by the same $r$-coloring if  $\{V_i,...V_j\} \cap \{V_s,...V_t\} = \phi$.
In order to compute the number of $r$-colorings used in the recursive step, we
 construct a graph $H$ as follows: 1. Put the elements of $\binom{[{\cal V}]}{p-1}$, each containing $(p-1)$ parts, as the vertices of $H$. 2. Connect an edge between two vertices if they share some part, i.e. if $\{V_i,...V_j\} \cap \{V_s,...V_t\} \neq \phi$, an edge is added between $\{V_i,...V_j\}$ and  $\{V_s,...V_t\}$. 
Consider some proper coloring of vertices of $H$. Each color class represents the set $S$ such that $x1 \in S$ and $x2 \in S$ implies $x1,x2 \in \binom{[{\cal V}]}{p-1}$ and $x1 \cap x2 = \phi$. Therefore, $\chi(H)$ $r$-colorings are sufficient in the recursive step. 
In order to compute $\chi(H)$, observe that $H$ is the line graph of the complete $p-1$-uniform $r$ vertex hypergraph
$H^L$ on the vertex set $\{V_1,...,V_r\}$.
So, using Baranayi's theorem \cite{Baranyai1979276,Berge197765},
$\chi(H)=\chi_e(H^L)=\lceil \frac{\binom{r}{p-1}}{\lfloor \frac{r}{p-1}\rfloor}\rceil$.
Let $D=\lceil \frac{\binom{r}{p-1}}{\lfloor \frac{r}{p-1}\rfloor}\rceil$.
Therefore, we need at most $D$ more $r$-colorings of $\frac{(p-1)n}{r}$ vertices to guarantee 
coverage of ${K_n}^k$.

We choose an integer  $l$, $0 \leq l \leq p-2$, such that $k-1-l$ is divisible by $p-1$.
To see that $T(n)=1$ for $n \leq \frac{r(k-1-l)}{p-1}$, observe that we can $r$-color the vertices  such that each color class is of size at most $\frac{k-1-l}{p-1}$, such that we need at least $p$ color classes to form a hyperedge of size $k$.
So we have the following recurrence.
\begin{align*}
T(n) \leq \begin{cases}
    1,\text{ if } n \leq \frac{r(k-1-l)}{p-1}, \text{ for }  0 \leq l \leq p-2,\\
    D T(\frac{(p-1)n}{r})+1,\text{ otherwise, where }D=\lceil \frac{\binom{r}{p-1}}{\lfloor \frac{r}{p-1}\rfloor}\rceil.
    \end{cases}
\end{align*}
Setting $\frac{(p-1)^i n}{r^i}=\frac{r(k-1-l)}{p-1}$, $i=\frac{\log (\frac{n}{k-1-l})}{\log(\frac{r}{p-1})}-1=\log_{\frac{r}{p-1}}\frac{n}{k-1-l}-1$. Consequently, 
\begin{align*}
T(n) \leq &D^{\log_{\frac{r}{p-1}}\frac{n}{k-1-l}-1}+\log_{\frac{r}{p-1}}\Big(\frac{n}{k-1-l}\Big)-1 & \\
= &\Big(\frac{n(p-1)}{(k-1-l)r}\Big)^{\log_{\frac{r}{p-1}}D}+\log_{\frac{r}{p-1}}\Big(\frac{n}{k-1-l}\Big)-1&.
\end{align*}
The following upper bound follows from the above inequality.
\begin{theorem}
Let $G(V,E)$ be an arbitrary $n$ vertex $k$-uniform hypergraph.
Let $r$ denote the number of colors used in each coloring of vertices of $V$, and $p$ denote the minimum number of distinctly colored vertices in any hyperedge to make it properly $(r,p)$ colored by any $r$-coloring.
Let $l$, $0 \leq l \leq p-2$, be fixed such that $k-l-1$ is divisible by $p-1$. Then,
\begin{align*}
\chi^c(G,k,r,p)\leq \chi^c(K_n^k,k,r,p) \leq \Big(\frac{n(p-1)}{(k-1-l)r}\Big)^{\log_{\frac{r}{p-1}}D}+\log_{\frac{r}{p-1}}{\Big(\frac{n}{k-1-l}\Big)}-1,
\end{align*}
where $D=\lceil \frac{\binom{r}{p-1}}{\lfloor \frac{r}{p-1}\rfloor}\rceil.$
\end{theorem}

\section{$p$-strong chromatic number and strong $(r,p)$ cover number}
\label{sec:semichromatic}
We know that the bicoloring cover number $\chi^c$ and the chromatic number $\chi$ are related as $\lfloor \log \chi(G) \rfloor \leq \chi^c(G) \leq \lceil \log \chi \rceil$ (\cite{tmspp2014}). In what follows, we investigate the relation between the strong $(r,p)$ cover number $\chi^c(G,k,r,p)$ and $p$-strong chromatic number $\chi(G,p)$.

Let $G(V,E)$ be a $n$ vertex $k$-uniform hypergraph. Let $C$ be a strong $(r,p)$ cover of size $\chi^c(G,k,r,p)$.
By definition of strong cover, every hyperedge is at least $p$-chromatic in at least one of the $r$-colorings in $C$. Every vertex has a color bit vector of $\chi^c(G,k,r,p)$ colors assigned to it by $C$, each color being one of $\{0,1,...,r-1\}$. 
Let $v \in V$ has the color bit vector $\{b_{\chi^c(G,k,r,p)-1},b_{\chi^c(G,k,r,p)-2},...,b_0\}$.
Consider the vertex coloring $X$ of the vertices in $V$  with $r^{\chi^c(G,k,r,p)}$ colors:
compute $c_v= r^{\chi^c(G,k,r,p)-1}b_{\chi^c(G,k,r,p)-1}+r^{\chi^c(G,k,r,p)-2}b_{\chi^c(G,k,r,p)-2}+...+r^0b_0$ and
assign $c_v$ as the color of $v$ in $X$.

We claim that $X$ is a valid strong $p$-coloring for $G$. Consider any hyperedge $e \in E(G)$. 
We know from definition of strong $(r,p)$ cover that there exists a
$r$-coloring $X_i \in C$ such that $e$ consists of at least $p$ distinctly colored vertices, $1 \leq i \leq \chi^c(G,k,r,p)$.
Let $x_1$,..., $x_p$ be the $p$ vertices of $e$ that are colored with distinct colors in $X_i$. From the computation of color $c_v$ for $v \in V$, it is clear that $x_1$,..., $x_p$ receive distinct colors in $X$. As a result, a valid strong $(r,p)$ cover of size $\chi^c(G,k,r,p)$ can be mapped to a valid  strong $p$-coloring with $r^{\chi^c(G,k,r,p)}$ colors. So, we have the following theorem.
\begin{theorem}
	For any $k$-uniform hypergraph $G(V,E)$,
	${\chi^c(G,k,r,p)} \geq \log_r \chi(G,p)$.
\end{theorem}

To check whether the converse also holds, we examine the case of $r=p=3$: given a valid strong 3-coloring with ${\chi(G,3)}$ colors, whether that can be mapped to strong $(3,3)$ cover of size $\log_3 {\chi(G,3)}$. We give a counter example to show that this may not hold for all hypergraphs. 
Consider the hypergraph $K_5^3$ with vertex set $\{v_1,...,v_5\}$.
Color each vertex $v_i, 1\leq i \leq 5$, with color $i$ to get a 5-coloring of vertices $X1$.
Note that $X1$ is a valid strong 3-coloring for $K_5^3$, so  ${\chi(K_5^3,3)} \leq 5$.
So, $\lceil \log_3 \chi(K_5^3,3) \rceil = 2$.
From Section \ref{subsec:exact}, we know that $\chi^c(K_5^3,3,3,3)$ is 3. Therefore, $\lceil \log_3 \chi(K_5^3,3) \rceil \not \geq \chi^c(K_5^3,3,3,3)$.


\section{Probabilistic analysis: relationship with the number of hyperedges and dependency}
\label{sec:probabilistic}
Let $m(r,p,x)$ denote the minimum number of hyperedges such that there exists a $k$-uniform hypergraph $G$ that does not have a strong $(r,p)$ cover of size $x$, $k \geq p$. We perform $x$ independent $r$-colorings on $G$. 
We start our analysis with $p=3$. Let the bad event $\mathcal{E}_i$ corresponds to the hyperedge $e_i$ colored with less than $3$ colors in each of the $x$ independent colorings.  
Observe that $\mathcal{E}_i$ can occur only if 
$e_i$ is either monochromatic or consists of exactly 2 colors in each of the $x$ independent bicolorings.
The probability (say $p1$) that $e_i$ is colored with exactly one color  in one $r$-coloring is $r(\frac{1}{r})^k=\frac{1}{r^{k-1}}$.
The probability that $e_i$ is colored with exactly two colors (say 0 and 1) in one $r$-coloring 
is $(\frac{2}{r})^k - 2(\frac{1}{r})^k$.
So, the probability (say $p2$)that $e_i$ is colored with exactly two colors (any one out of $\binom{r}{2}$ pairs) in one $r$-coloring is $\binom{r}{2}((\frac{2}{r})^k - 2(\frac{1}{r})^k)$.
Therefore, sum of $p2$ and $p1$ is the probability $p3$ of the event that $e_i$ colored with less than 3 colors in one $r$-coloring.
The probability that $e_i$ is colored with less than 3 colors in each of the
$x$ independent $r$-colorings is $(p3)^x$, which is equivalent to $p(\mathcal{E}_i)$.
So, 
\begin{align}
p(\mathcal{E}_i)&=\Bigg(\binom{r}{2}\Big(\big(\frac{2}{r}\big)^k - 2\big(\frac{1}{r}\big)^k \Big)+\frac{1}{r^{k-1}}\Bigg)^x  \nonumber \\
&=\Bigg(\binom{r}{2}(\frac{2}{r})^k - r(r-1)(\frac{1}{r})^k+r\frac{1}{r^{k}}\Bigg)^x \nonumber\\
&=\Bigg(\binom{r}{2}(\frac{2}{r})^k - r(r-2)(\frac{1}{r})^k \Bigg)^x \nonumber\\
& \leq \Bigg(\binom{r}{2}(\frac{2}{r})^k \Bigg)^x.
\end{align}
%
%
If the number of hyperedges is less than or equal to $ \frac{1}{2} \big(\frac{1}{r-1}(\frac{r}{2})^{(k-1)}\big)^x$, summing the bad events $p({\cal{E}}_i)$ for each $e_i \in E(G)$, the probability that any hyperedge is colored with less than 3 colors in each of the $x$ $r$-colorings is less than $\frac{1}{2}$. Consequently, the hypergraph has a strong $(r,3)$ cover of size $x$ and that can be obtained by $x$ random $r$-colorings of vertices in expected two iterations. 

For the case of strong $(r,p)$ cover, the bad event $\cal{E}$ corresponds to the hyperedge $e$ colored with less than $p$ colors in each of the $x$ independent colorings. The probability $p({\cal{E}})$ is at most $(\binom{r}{p-1}{(\frac{p-1}{r})}^{k})^x$: the hyperedges colored with less than $p-1$ colors are counted more than once
in this probability.
 In order to enforce the condition that $p({\cal{E}})$ is less than 1, we need to ensure $\binom{r}{p-1}{(\frac{p-1}{r})}^{k}$ is less than 1.
\begin{align*}
&\binom{r}{p-1}{\Big(\frac{p-1}{r}\Big)}^{k} \leq \Big(\frac{er}{p-1}\Big)^{(p-1)}{\Big(\frac{p-1}{r}\Big)}^{k}
= e^{p-1}{\Big(\frac{p-1}{r}\Big)}^{(k-p+1)} < 1&\\
=> &p-1+(k-p+1)(\ln (p-1)- \ln r) < 0&
\end{align*}
Choosing a value of $k \geq 2p-1$, $r \geq e(p-1)$, this condition is satisfied. If the number of hyperedges is less than or equal to $\frac{1}{2}({\frac{r^k}{(p-1)^k \binom{r}{p-1}}})^x$, summing the probabilities for each hyperedge, the probability that any hyperedge is colored with less than $p$ colors in each of the $x$ $r$-colorings is less than $\frac{1}{2}$. Therefore, $\frac{1}{2}({\frac{r^k}{(p-1)^k \binom{r}{p-1}}})^x < m(r,p,x)$, provided $k \geq 2p-1$, $r \geq e(p-1)$.



The {\it dependency} of a hyperedge $e$ in a $k$-uniform 
hypergraph $G(V,E)$,
denoted by $d(G,e)$ 
is the number of hyperedges in the set $E$ 
with which $e$ shares at least one vertex. 
The {\it dependency of a hypergraph} $d(G)$ or simply $d$, denotes the 
maximum dependency of any hyperedge in the hypergraph $G$.
Lov\'{a}sz local lemma  
\cite{loverd1975,Motwani:1995:RA:211390} 
 and its constructive
version 
\cite{Moser:2010:CPG:1667053.1667060} 
ensures the existence and enables the computation of a bicoloring of a $k$-uniform 
hypergraph with dependency at most   
$\frac{2^{k-1}}{e}-1$, respectively. 
Chandrasekaran et.al. \cite{Chandrasekaran:2010} proposed a derandomization  for local lemma  that computes a bicoloring
in polynomial time. In what follows, we  
use similar techniques for establishing permissible bounds on the 
dependency of a hypergraph as a function of the size of its 
desired strong $(r,p)$ cover, and for computing such strong $(r,p)$  covers.

Let the bad event ${\cal E}_i$ correspond to the hyperedge $e_i$ colored with less than $p$ colors in each of the $x$ independent $r$-colorings. 
${\cal E}$ be the set of all the bad events.
The probability $p({\cal E}_i)$ is at most $(\binom{r}{p-1}{(\frac{p-1}{r})}^{k})^x$.
So, by the direct application of local lemma corollary, the maximum allowable dependency $d$ of a hyperedge becomes
$\frac{1}{e}({\frac{r^k}{(p-1)^k \binom{r}{p-1}}})^x-1$. In other words, if the dependency of the hypergraph $G$ is less than or equal to $\frac{1}{e}({\frac{r^k}{(p-1)^k \binom{r}{p-1}}})^x-1$, then $\chi^c(G,k,r,p) \leq x$.
Moreover, a strong $(r,p)$ cover of size $x$ can be computed in randomized polynomial time using an adaption of Moser-Tardos algorithm, \ref{algo:admostar}.

\begin{algorithm}[H]
\SetAlgoRefName{$MTC$}
 \KwData{$k$-uniform hypergraph $G(V,E)$ with $d\leq \frac{1}{e}({\frac{r^k}{(p-1)^k \binom{r}{p-1}}})^x-1$}
 \KwResult{Set $X$ of $r$-colorings of size $x$ }
\For{$v \in V$}{
	\For{$i \in \{1,...,x\}$}{
			${r_v}^i \leftarrow$ a random evaluation of $v$ in $i^{th}$ $r$-coloring of $X$\;
		}
}
\While{$\exists {\cal E}_i \in \mathcal{E}$: ${\cal{E}}_i$ happens i.e., every $r$-coloring in $X$ colors $e_i$ with less than $p$ colors}{
		Pick an arbitrary violated event ${\cal E}_i \in \mathcal{E}$\;
		\For{$ v \in e_i$}{
			\For{$i \in \{1,...,x\}$}{
					${r_v}^i \leftarrow$ a random evaluation of $v$ in $i^{th}$ $r$-coloring  of $X$\;
				}
		}
}
\caption{Randomized algorithm for computing a strong $(r,p)$ cover}
\label{algo:admostar}
\end{algorithm}


\bibliographystyle{plain}
   		\bibliography{References}
\end{document}